\documentclass[lettersize,journal]{IEEEtran}
\usepackage[utf8]{inputenc} 
\usepackage[T1]{fontenc}
\usepackage{url}
\usepackage{ifthen}
\usepackage{cite}
\usepackage[cmex10]{amsmath} 
\usepackage{graphicx}
\usepackage{amsmath,amssymb,amsthm}
\usepackage{cite,soul}
\usepackage{tikz}
\usepackage[capitalize]{cleveref}
\usepackage{diagbox}
\usepackage{multirow}
\usepackage{rotating}
\usepackage{threeparttable}
\usepackage{booktabs}
\usepackage{comment}
\usepackage{enumitem}
\usepackage{algorithm}
\usepackage{algpseudocode}
\usepackage{multirow,rotating}
\DeclareMathOperator{\parD}{Pareto}

\DeclareMathOperator{\expD}{Exp}
\DeclareMathOperator{\sexpD}{S-Exp}
\DeclareMathOperator{\expec}{\mathbb{E}}

\newtheorem{theorem}{Theorem}
\newtheorem{lemma}{Lemma}

\def\BibTeX{{\rm B\kern-.05em{\sc i\kern-.025em b}\kern-.08em
    T\kern-.1667em\lower.7ex\hbox{E}\kern-.125emX}}

\newcommand{\prob}[1]{{\mathbb P}}
\begin{document}
\title{Redundancy Management for Fast Service (Rates) in Edge Computing Systems}

\author{Pei Peng, ~\IEEEmembership{Member,~IEEE}, Emina Soljanin,~\IEEEmembership{Fellow,~IEEE}
\thanks{This work was supported in part by the National Natural Science Foundation of China (Grant No. 62301273),  the Natural Science Foundation of Jiangsu Province (Grant No. BK20230355), the University Science Research Project of Jiangsu Province (Grant No. 23KJB120009), the Natural Science Research Start-up Foundation of Recruiting Talents of Nanjing University of Posts and
Telecommunications (Grant No. NY222008) and the NSF-BSF Award FET-2120262. This article was presented in part at the 2023 15th International Conference on Wireless Communications and Signal Processing (WCSP)\cite{peng_newaccepted} and in part at the 2023 IEEE 23rd International Conference on Communication Technology (ICCT)\cite{peng_newaccepted_ICCT}.}
\thanks{Pei Peng is with the School of Telecommunication and Information Engineering, Nanjing University of Posts and Telecommunications, Nanjing, Jiangsu 210003, China (email: pei.peng@njupt.edu.cn).}
\thanks{Emina Soljanin is with the Department of Electrical and Computer Engineering, Rutgers, The State University of
New Jersey, Piscataway, NJ 08854, USA (e-mail: emina.soljanin@rutgers.edu)}}

\markboth{IEEE/ACM Transactions on Networking, to appear}%
{Shell \MakeLowercase{\textit{et al.}}: A Sample Article Using IEEEtran.cls for IEEE Journals}

\IEEEpubid{0000--0000/00\$00.00~\copyright~2021 IEEE}
\IEEEpubidadjcol 

\maketitle

\begin{abstract}
Edge computing operates between the cloud and end users and strives to provide low-latency computing services for simultaneous users. Redundant use of multiple edge nodes can reduce latency, as edge systems often operate in uncertain environments. However, since edge systems have limited computing and storage resources, directing more resources to some computing jobs will either block the execution of others or pass their execution to the cloud, thus increasing latency. This paper uses the average system computing time and blocking probability to evaluate edge system performance and analyzes the optimal resource allocation accordingly. We also propose blocking probability and average system time optimization algorithms. Simulation results show that both algorithms significantly outperform the benchmark for different service time distributions and show how the optimal replication factor changes with varying parameters of the system. 
\end{abstract}

\section{Introduction}
\IEEEPARstart{W}{ith} the rapid increase in IoT applications, such as smart cities and homes, autonomous vehicles, and artificial intelligence, billions of IoT devices are coming into our everyday lives \cite{shi2016edge,mao2017survey}.
The demand for low latency storage and computing services is increasing to accommodate novel IoT platforms (e.g., deep learning) \cite{li2018learning,hochstetler2018embedded,li2021random}. Some applications, for example, connected and autonomous vehicles, smart healthcare, and ocean monitoring require fast or no service.
Cloud services are inefficient at responding to such applications. 

Edge computing is an inter-layer between the cloud and the end-user. It provides storage and computing infrastructure at the node located one or two network hops from the end-user \cite{byers2017architectural}. According to \cite{yi2015fog}, the round trip time between the cloud and end-user (about $17.989$ ms) is over $10$ times larger than the time between the edge and end-user (about $1.416$ ms). Therefore, in the edge system, the bottleneck of the computing service is no longer the communication latency.
As illustrated in Figure~\ref{fig:cloud}, an edge computing system receives and processes jobs from end-users on a much smaller scale of storage and computing resources than the cloud \cite{chiang2016fog,mahmud2018fog}. Service requests get sent to the cloud when all edge workers are busy, and communication time becomes a significant part of the latency.
\begin{figure}[hbt]
\begin{center}
\begin{tikzpicture}
\node at (2.2,0) {\includegraphics[width=0.4\textwidth]{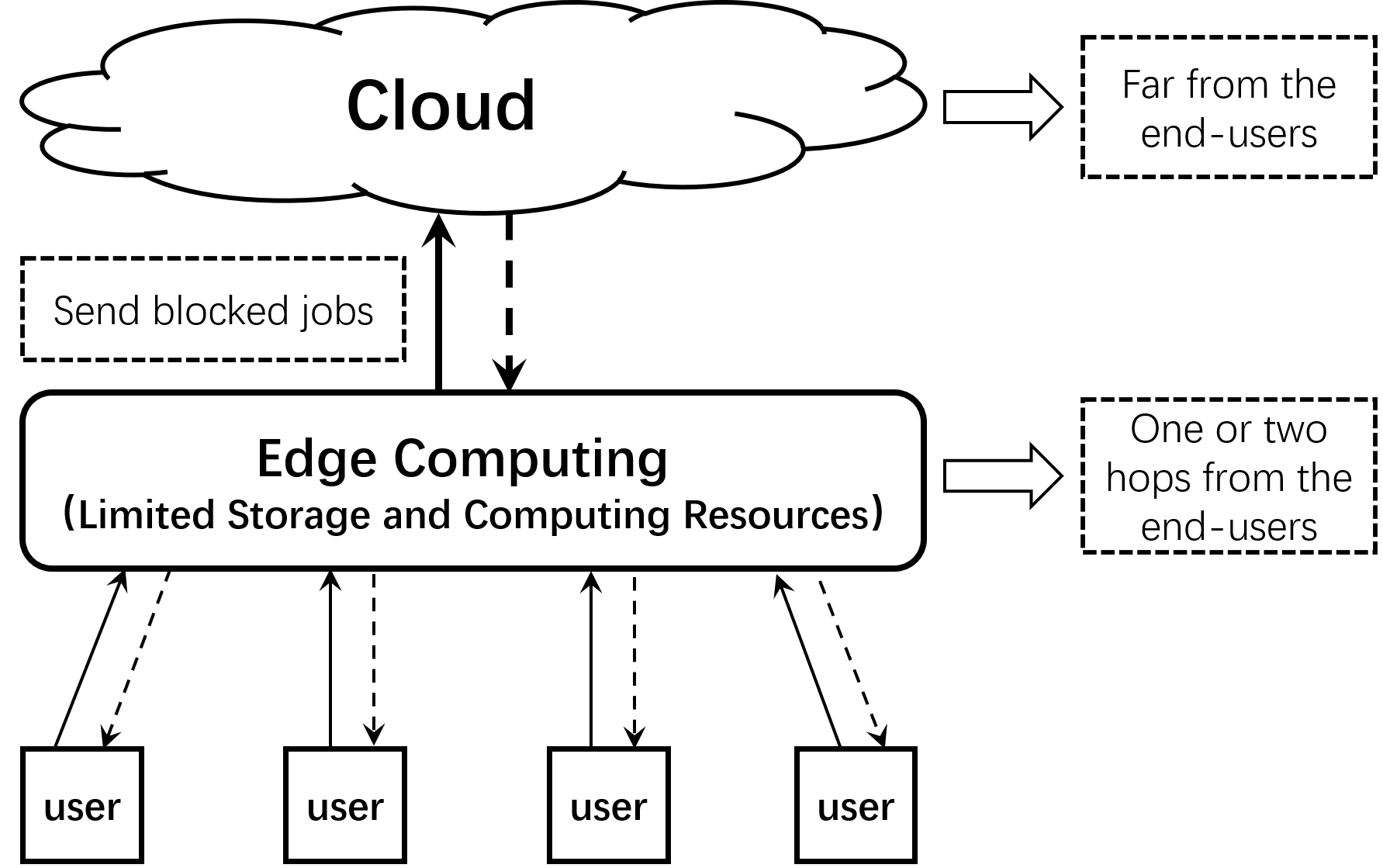}};
\end{tikzpicture}
\end{center}
\vspace*{-3ex}
\caption{Edge computing system deployed between the users and the cloud processes the jobs sent by the users. New service requests get sent to the cloud when all edge workers are busy.
\label{fig:cloud}}
\end{figure}

Edge nodes often operate in highly unpredictable environments and handle jobs that need fast or no service. Redundant use of resources is a known strategy to enable fast service. However, with increased resource usage comes a decrease in the number of jobs that the edge system can execute simultaneously. Thus, more jobs are sent to the cloud, increasing their execution time. An edge computing system that sends jobs to the cloud once resources become unavailable acts locally as a blocking system. The jobs sent to the cloud will experience higher latency as they forgo the geographic benefit \cite{satyanarayanan2017emergence}. 
\IEEEpubidadjcol

We, therefore, need to address two opposing objectives to maximize the benefits of edge computing. The first objective is to minimize the computing speed of each job executed locally at the edge. (Computing speed, expressed in terms of computing time, is a classical and essential performance metric \cite{wang2019efficient,lee2017speeding}.) The second objective is to minimize the number of jobs the edge system processes (i.e., the number of blocked jobs sent to the cloud). Due to limited storage and computing resources, the edge system may be unable to process all jobs independently of the cloud. 

There are many theoretical and practical models to solve open problems in edge computing\cite{kolosov2023theory}, we here only focus on one of them. We consider edge computing systems in which a single controller node manages a cluster of workers in computing \cite{mouradian2017comprehensive}. When a job arrives, the controller replicates it to several workers. Here, we refer to the workers processing the same job as the replication group and its size as the replication factor. Since the total number of workers in the system is limited, we have to decrease the number of groups when we increase the replication factor. The edge computing system with a higher replication factor may have a higher job-blocking probability. That is, it can process fewer jobs simultaneously. 
Therefore, processing more jobs locally may decrease the expected computing speed. 

It is generally impossible to simultaneously minimize the local job-computing time and job-blocking probability. Since both are crucial to edge computing, we will focus on finding the trade-off between these performance metrics. We combine the local edge computing time and the time spent reaching the cloud (which happens with the job-blocking probability at the edge) into a single metric representing the average service time. 
We would like to remind the reader that other related crucial edge performance metrics are beyond the scope of this paper. For example, \cite{saleem2020mobility} focuses on minimizing task offload latency by jointly formulating task assignment and power allocation.\cite{liu2019dynamic} proposed a novel offloading system with resource allocation strategies to achieve minimum energy. \cite{zhu2018task} formulate different energy minimization problems by exploiting the user's mobility prediction. \cite{liu2016delay} focuses on minimizing task delay and power consumption by proposing an efficient one-dimensional search algorithm.

A large body of literature argues the benefits of redundancy in distributed computing, e.g., \cite{kianidehkordi2020hierarchical,CC:ozfaturayGU19,schlegel2023codedpaddedfl}, and references therein. Redundancy also benefits the edge systems. For example, \cite{zhang2019model} showed that coding is instrumental in reducing overall computation and communication latency. \cite{wang2021optimal} proposed the linear coding policy to minimize resource consumption with heterogeneous mobile devices. \cite{choudhury2022tackling} coded across the servers to flexibly serve multiple types of jobs without additional storage cost. However, because of the adverse effects of temporarily redundantly using more system resources, there are limits to these benefits. The optimal levels of redundancy have yet to be fully characterized; see, e.g., \cite{stragglers:AktasS19,peng2021diversity} and references therein, and are the subject of current research. In which, \cite{stragglers:AktasS19} analyzed the cost and latency problem affected by the replicated or erasure-coded redundancy in distributed systems; \cite{peng2021diversity} analyzed the level of redundancy that affects the computation time under different service time distributions and task size scaling models.

The contributions of this paper are summarized as follows:
\begin{enumerate}[leftmargin=*]
\item We propose an edge computing system model that combines the distributed and blocking systems. We aim to evaluate the appropriate performance metrics and find the optimal replication factor under different service time distributions. 

\item We theoretically analyze the job-blocking probability for shifted exponential and Pareto service times and conclude spreading computing resources to more jobs is a better strategy. Meanwhile, we find the optimal replication factor varies with the tail index for Pareto service time and propose the BPO algorithm accordingly.   

\item We adopt {\it the average system time} to evaluate the trade-off between the job-computing time and the system job-blocking probability. The theoretical and numerical analysis demonstrates that different values of replication factors affect system performance. Then, we propose the ASTO algorithm accordingly, and the simulation results verify the effectiveness of the redundant computing resource allocation.

\item We separately analyze the minimum average system time and the optimal replication factor, which changes with different system parameters such as the job arrival rate, the cloud time, and the number of workers. We summarize the patterns of the variation of the optimal allocation strategy.
\end{enumerate}

The paper is organized as follows. In Section~\ref{sec:sys}, we describe the architecture of the edge computing system and formulate the problem. In Section~\ref{sec:problem}, we state the problem and introduce the formulas of performance metrics. In Section~\ref{sec:ave_exp}, we theoretically and numerically analyze the job-blocking probability and propose the BPO algorithm for Pareto service time.
In Section~\ref{sec:ave_shift}, we find the optimal replication factor for the average system time under certain conditions and propose the ASTO algorithm. In Section~\ref{Sec:para}, we analyze how the average system time and the optimal replication factor change with different system parameters.
Conclusions are given in Section~\ref{sec:conclusions}. 

\section{System Model}
\label{sec:sys}

\begin{figure}[hbt]
\begin{center}
\begin{tikzpicture}
\node at (2.2,0) {\includegraphics[width=0.475\textwidth]{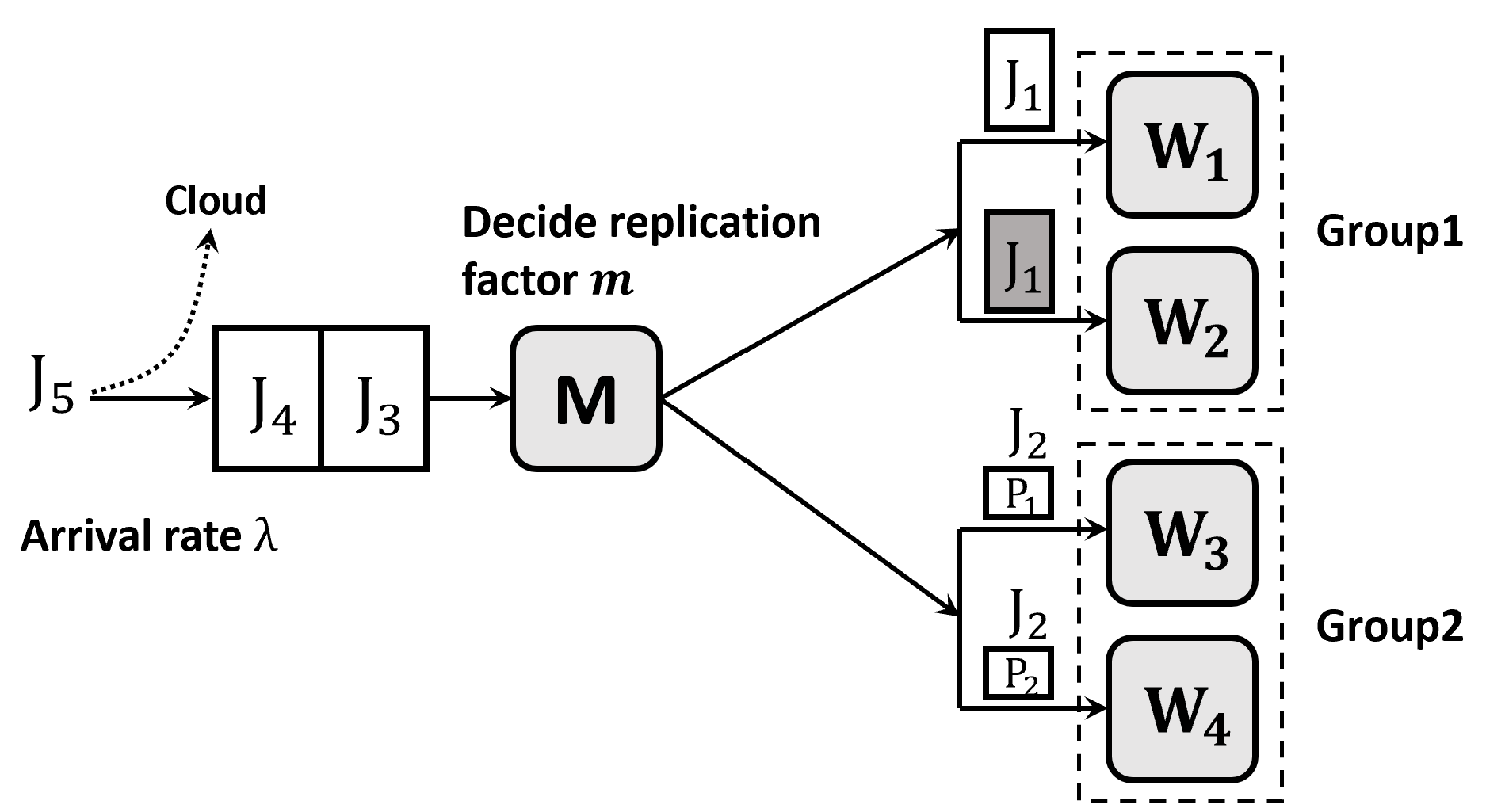}};
\end{tikzpicture}
\end{center}
\vspace*{-3ex}
\caption{\textbf{An Edge Computing System:} Controller node $M$ processes jobs $J_i$, possibly generates redundant tasks and dispatches them to workers $W_1, W_2 , W_3, W_4$. $W_1$ and $W_2$ work as a group to process $J_1$, where the shaded job of $J_1$ indicates a redundant job. $W_3$ and $W_4$ also work as a group; each worker processes $1/2$ part of $J_2$.
\label{fig:model}}
\end{figure}

\subsection{System Architecture}

We consider the edge computing system model shown in Figure~\ref{fig:model} as a combination of the distributed computing system and the blocking system. The edge computing system has limited storage and computing resources\cite{mao2017survey}. It consists of a single front-end controller node and multiple computing servers, which we refer to as workers. The single controller node manages the entire computing cluster of nodes. This architecture is commonly implemented in modern frameworks, such as Apache Mesos~\cite{Mesos:HindmanKZ11}, and edge computing systems with limited storage and computing resources~\cite{mouradian2017comprehensive}. Besides, this model is not limited to the above distributed structured system. In the edge computing system, the computing cycles in the CPU are practical to evaluate the amount of computing resources\cite{luo2021resource}. Each job requests specific computing cycles, but the edge system can simultaneously offer resources to limited jobs. Thus, when we define the worker as the number of request computing cycles for one job, the edge computing system can be modeled as in Figure~\ref{fig:model}. 

The controller node will divide the workers into several groups. The controller node creates $m$ copies for each arriving job and assigns each copy to a worker in a group. In Figure~\ref{fig:model}, the controller node sends $J_1$ and its copy to workers 1 and 2 ($m=2$) and sends $J_2$ and its copy to workers 3 and 4 ($m=2$).
This execution replication mitigates straggling and reduces the job's expected completion time. The larger the replication factor $m$, the higher the job's average completion time reduction.

On the other hand, when all workers are busy, the new request for job execution gets blocked. An edge computing system may send such jobs to the cloud (see Figure~\ref{fig:model}, job $J_5$). There is a significant communication latency between the edge and the cloud, which may be much longer than the expected job-computing time. 
Therefore, we want the system to serve more jobs and send fewer to the cloud.

\subsection{Redundancy and Straggler Mitigation}

job-computing time $T_{\text{job}}$ is a crucial performance metric in an edge computing system. However,
task straggling is a fundamental problem in distributed systems that significantly affects performance. To solve this problem, replication is an effective technique to introduce redundancy to mitigate stragglers. 
In Figure~\ref{fig:model}, the controller node sends the $J_1$ and its copy separately to workers 1 and 2. Compared to $J_2$ without redundancy, even though one worker processes the job for a long time, the controller node can still receive the result from the other worker. This problem is well studied in \cite{peng2021diversity}.

\subsection{Performance Metrics}

There are two performance metrics of interest in the described system: 1) the job-computing time $T_{\text{job}}$ and 2) the job-blocking probability $P_b$. The first metric focuses on processing each job fast, and the second focuses on processing more jobs in the edge system.
The system's goal is to minimize both metrics. The design parameters for a given system size (fixed number of workers) are the replication factor $m$ and the number of groups $c$. Increasing $ m$ and decreasing $ c$ reduce $ T_textjob$ (improving the first performance metric). However, the effect of $m$ and $c$ on $P_b$ is not immediately apparent. Increasing $m$ will temporarily occupy more servers per job but will also make the jobs stay in the system for a shorter time. 

In general, there may not be an optimal replication factor $m$ that simultaneously minimizes both $\expec[T_{\text{job}}]$ and $P_b$. 
We use the \textit{ average system time} to evaluate these two metrics' trade-offs. The average system time is defined as 
\begin{equation}
    \expec[T_{\text{sys}}]=(1-P_b)\expec[T_{\text{job}}]+P_b\expec[T_{\text{cl}}].
    \label{Eq:avetime}
\end{equation}
where $T_{\text{cl}}$ is the completion time of the jobs blocked by the edge system and executed in the cloud.

We consider the average system time as the primary performance metric to analyze in this paper. 
According to the system model, the job-computing time is generally smaller than the cloud time. Otherwise, the users should send all jobs to the cloud. Therefore, the average system time focuses more on reducing job-blocking probability when the cloud time is larger and reducing the job-computing time when the cloud time is smaller. 

\subsection{Job Arrival and Service Time}

M/M/1 or M/M/n queues have been used in task allocation models in edge systems\cite{chen2021recent,fan2018application,chen2018computation}. Here, we assume the job arrivals follow a Poisson process with a rate $\lambda$, which allows us to model the blocking system as an M/G/c/k queue. We use $c$ to evaluate the system capacity, defined as the number of jobs the system can process simultaneously. When $c$ is small, the blocking system is modeled as an edge system with relatively limited computing resources or processes large-size jobs, such as edge computing in the low-earth orbit satellite environment\cite{tang2021computation,wang2022collaborative}. Otherwise, the blocking system is modeled as a general edge system.

Analyzing the job-blocking probability of an M/G/c/k queue that models the blocking operation of the system in Figure~\ref{fig:model} is a highly complex problem. To better understand it, we consider the M/G/c/c queue, where the queue length equals the number of groups. 
Then, the service time follows the shifted exponential distribution $\sexpD(\Delta,\mu)$, where $\Delta$ is the shift parameter and $\mu$ is the rate parameter, a straightforward model widely used in distributed computing\cite{latency:JoshiSW17,bitar2020minimizing,CC:duttaCP16}. Here, $\Delta$ is an initial handshake time, after which the worker will complete the job in some $\expD(\mu)$ time. Notice that the exponential service time is a particular case when $\Delta=0$, the blocking system becomes an Erlang B model.
Similarly, we will also consider the M/G/c/c queue with the $\parD(\beta,\alpha)$ service time, where $\beta$ is the scale parameter and $\alpha$ is the tail index. Pareto distribution is also widely used in distributed computing\cite{stragglers:AktasS19,GradientCoding:TandonLD17}.

\section{Problem Statement and Formulas}
\label{sec:problem}
\subsection{Problem Statement}
\begin{center}
    \begin{tabular}{ rl}
       $N$  --  &  number of workers in the system\\
       $m$   --  & replication factor, number of workers in a group\\
       $c$  -- & number of groups, $c=N/m$\\
       $\lambda$ -- & job arrival rate\\
       $T_{\text{cl}}$ -- & time that a job spends in the cloud (cloud time)\\
       $T_{\text{job}}$ -- & job-computing time \\
       $P_b$ -- & job-blocking probability \\
       $T_{\text{sys}}$ -- & average system time\\
    \end{tabular}
\end{center}

The system parameters and notations are summarized in the above list.
Our goal is to evaluate two performance metrics, the expected job-computing time $\expec[T_{\text{job}}]$ and the job-blocking probability $P_b$ for systems with Poisson job arrivals and shifted exponential/Pareto service time. We separately compute the replication factor $m$ (or the number of groups $c$) that minimizes $\expec[T_{\text{job}}]$ and $P_b$. When $c$ is small, the system concentrates limited computing resources on a few jobs. When $c$ is large, the system spreads the computing resources over more jobs. The edge system will process a job faster with more computing resources and slower with fewer resources. To better evaluate $\expec[T_{\text{job}}]$ and $P_b$, we use the average system time $\expec[T_{\text{sys}}]$ to evaluate the trade-off between these two metrics.
We show how to achieve the desired trade-off between these two traditional metrics by selecting an appropriate $m$. 

We summarize the optimal resource allocation on the average system time in Table~\ref{Table:ave}. We conclude that it is better to concentrate computing resources on a few jobs when the job arrival rate and the cloud time are small. Otherwise, we should consider spreading computing resources to more jobs. Generally, we should balance concentrating and spreading the computing resources as the increase of the number of workers. The tail index is also essential in selecting the optimal resource allocation for Pareto service time.

\begin{table*}[hbt]
\begin{center}
\begin{threeparttable}
    \caption{Optimal Resource Allocation on Average System Time }
    \label{Table:ave}
				\begin{tabular}{@{}lcc|cc|cc|cc@{}}
				\toprule 
				 &  \multicolumn{2}{c|}{\sc Arrival Rate $\lambda$} &  \multicolumn{2}{c|@{}}{\sc Cloud Time $\expec[T_{\text{cl}}]$} &  \multicolumn{2}{c|@{}}{\sc Number of Workers $N$} & \multicolumn{2}{c|@{}}{\sc Tail Index $\alpha$} 
				\\\cmidrule{2-9}
				 &{\bf Large}  & {\bf Small} &  {\bf Large}  & {\bf Small} & {\bf Large}  & {\bf Small} & {\bf Large}  & {\bf Small}\\ [2mm]
			   {\bf Exp}  &  B & C  &  S & C& C & B& \textbackslash  & \textbackslash \\[2mm] 
				{\bf S-Exp}   & B & C   & S  & C & B  & B &\textbackslash  & \textbackslash \\[2mm]
    {\bf Pareto}   & B & C   & B  & C & B  & B & C  & S \\
				\bottomrule
			\end{tabular}
\begin{tablenotes}
\footnotesize
\item C means concentrating computing resources on a few jobs.
\item S means spreading computing resources to more jobs.
\item B means balance between concentrating and spreading computing resources
\end{tablenotes}
    \end{threeparttable}
	\end{center}
\end{table*}

We here put together the relevant results of \cite{peng_newaccepted_ICCT} and \cite{peng_newaccepted}, which focuses on analyzing the average system time and the system service rate in the edge system under exponential service time (part of Section~\ref{sec:ave_exp})).

\subsection{Formulas of Performance Metrics}
\label{subsec:formula}
We adopt the shifted-exponential distributions as classical and simple service time models to analyze the edge computing system. We consider the system to use $m$-fold replication. Considering $\sexpD(\Delta,\mu)$ service time, the expected job-computing time is given by 
\begin{equation}
\label{Eq:time_shift}
    \expec[T_{\text{job}}(m)]= \Delta+\frac{1}{m\mu}\; \;\text{or}\; \;\expec[T_{\text{job}}(c)]=\Delta+\frac{c}{N\mu}.
\end{equation}
Observe that $T_{\text{job}}$ follows the shifted exponential distribution with the shift $\Delta$ and the rate parameter $m\mu$.  The edge system will spend $\Delta+\frac{1}{\mu}$ time processing the job with the minimum computing resource and $\Delta+\frac{1}{N\mu}$ time to process the job with the maximum computing resource.  

Considering $\parD(\beta,\alpha)$ service time,  the expected job-computing time is given by 
\begin{equation}
\label{Eq:time_pareto}
    \expec[T_{\text{job}}(m)]= \frac{\beta m \alpha}{m\alpha-1}\; \;\text{or}\; \;\expec[T_{\text{job}}(c)]=\frac{\beta N \alpha}{N\alpha-c}.
\end{equation}
The edge system will spend $\frac{\beta\alpha}{\alpha-1}$ time processing the job with the minimum computing resource, and $\frac{\beta\alpha N}{\alpha N-1}$ time to process the job with the maximum computing resource.

Analyzing the job-blocking probability of an M/M/c/k (or M/G/c/k) queue that models the blocking operation of the system in Figure~\ref{fig:model} is a highly complex problem. To better understand it, we consider the M/G/c/c queue, where the queue length equals the number of groups. Then, for a blocking system with $c$ groups and the job arrives as a Poisson process with the rate $\lambda$, the job-blocking probability is 
\begin{equation}
\label{Eq:blocking}
P_b(c,\rho)=\frac{(c\rho)^c/c!}{\sum_{j=0}^{c}(c\rho)^j/j!}.
\end{equation}
Where $\rho=\lambda\expec[T_{\text{job}}]/c\le1$. The above expression shows that for a given $c$, $P_b$ increases with $\rho$.

\section{job-blocking probability}
\label{sec:ave_exp}
We first consider processing more jobs in the edge computing system to use the advantage of edge computing. When computing resources are concentrated on a few jobs, the edge system will process each job faster by allocating more workers and using replication redundancy. But, the edge system can only process fewer jobs simultaneously. In the following, we will analyze how the replication factor $m$ and the number of groups $c$ affect the system performance regarding the job-blocking probability $P_b$.

\subsection{Shifted Exponential Service Time}

When the worker computes each job follows $\sexpD(\Delta,\mu)$, the expected job-computing time is $\expec[T_{\text{job}}(c)]=\Delta+\frac{c}{N\mu}$ according to \eqref{Eq:time_shift}, then we can rewrite \eqref{Eq:blocking} in the following,
\begin{equation}
\label{Eq:rewrite_shift}
P_b(c)=\frac{(K(\Delta N\mu+c))^{c}/c!}{\sum_{j=0}^{c}(K(\Delta N\mu+c))^j/j!}.
\end{equation}
Where $K=\frac{\lambda}{N\mu}$.
From the above formula, we find the optimal replication factor $m$ for the job-blocking probability, as follows.

\begin{theorem}
\label{Thm:block}
For the blocking system with Poisson($\lambda$) arrivals and $\sexpD(\Delta,m\mu)$ service time, the job-blocking probability $P_b$ decreases with increasing $c$ and reaches the minimum at $m=1$ (i.e., $c=N$).
\end{theorem}
\begin{proof}

Assume $c_1>c_0$, from  \eqref{Eq:rewrite_shift}, the job-blocking probability is 
\begin{align*}
&P_b(c_1)=\frac{(K(\Delta N\mu+c_1))^{c_1}/c_1!}{\sum_{j=0}^{c_1}(K(\Delta N\mu+c_1))^j/j!}\\
&=\frac{1}{\sum_{j=0}^{c_1}\frac{c_1!/j!}{(K(\Delta N\mu+c_1))^{c_1-j}}}
=\frac{1}{\sum_{j=0}^{c_1}(\prod_{i=1}^{j}\frac{c_1-j+i}{c_1+\Delta N\mu})\frac{1}{K^j}}\\
&<\frac{1}{\sum_{j=0}^{c_0}(\prod_{i=1}^{j}\frac{c_1-j+i}{c_1+\Delta N\mu})\frac{1}{K^j}}<\frac{1}{\sum_{j=0}^{c_0}(\prod_{i=1}^{j}\frac{c_0-j+i}{c_0+\Delta N\mu})\frac{1}{K^j}}\\
&=\frac{(K(\Delta N\mu+c_0))^c_0/c_0!}{\sum_{j=0}^{c_0}(K(\Delta N\mu+c_0))^j/j!}=P_b(c_0).
\end{align*}
Since $c\in[1,N]$, $P_b$ increases with the number of groups $c$ and reaches its minimum at $c=N$.
\end{proof}

From Theorem~\ref{Thm:block}, we conclude that spreading computing resources over more jobs is a better way to support the system processing more jobs. We have similar results for exponential service time to the above theorem since exponential distribution is a special case of the shifted exponential distribution when the shift parameter $\Delta=0$.

\subsection{Pareto Service Time}
When the worker computes each job follows $\parD(\beta,\alpha)$, the expected job-computing time is $\expec[T_{\text{job}}(m)]=\frac{\beta m\alpha}{m\alpha-1}$ according to \eqref{Eq:time_pareto}, then we can rewrite \eqref{Eq:blocking} as a function of $m$ in the following,
\begin{equation}
\label{Eq:rewrite_pareto}
P_b(m)=\frac{(\lambda\expec[T_{\text{job}}(m)]m/N)^{N/m}/(N/m)!}{\sum_{j=0}^{N/m}(\lambda\expec[T_{\text{job}}](m)m/N])^j/j!}.
\end{equation}
From the above formula, we find the optimal replication factor $m$ for the job-blocking probability, as follows.

\begin{lemma}
    \label{Le:block_pareto}
    For the blocking system with Poisson($\lambda$) arrivals and $\parD(\beta,\alpha)$ service time, when $\alpha\ge 1.5$, the job-blocking probability $P_b$ increases with $m$ and reaches the minimum at $m=1$ (i.e., $c=N$); When $\alpha<1.5$, $P_b$ reaches the minimum at $\min\{P_b(m=1), P_b(m=2)\}$.
\end{lemma}
\begin{proof}

From Theorem~\ref{Thm:block}, we have $P_b(c)>P_b(c+1)$ also holds for exponential service time. Since $m=\frac{N}{c}$, we have $P_b(m)<P_b(m+1)$. Let $r_{exp}=\frac{\expec[T_j(m+1)]}{\expec[T_j(m)]}=\frac{m}{m+1}$. Let $T_{par}(m)$ be the job-computing time for Pareto service time and $r_{par}=\frac{\expec[T_{par}(m+1)]}{\expec[T_{par}(m)]}$. 

For Pareto service time, according to the above conclusion of exponential service time, we have the relationship of the job-blocking probability $P_{par}(m)$ when $r_{par}\ge r_{exp}$.
\begin{align*}
&P_{par}(m+1)\\&=\frac{(\lambda \expec[T_{par}(m+1)](m+1)/N)^{N/(m+1)}/(N/(m+1))!}{\sum_{j=0}^{N/(m+1)}(\lambda \expec[T_{par}(m+1)](m+1)/N)^j/j!}\\
&>\frac{(\lambda \expec[T_{par}(m+1)]m/(r_{exp}N))^{N/m}/(N/m)!}{\sum_{j=0}^{N/m}(\lambda \expec[T_{par}(m+1)]m/(r_{exp}N)^j/j!}\\
&\ge\frac{(\lambda \expec[T_{par}(m+1)]m/(r_{par}N))^{N/m}/(N/m)!}{\sum_{j=0}^{N/m}(\lambda \expec[T_{par}(m+1)]m/(r_{par}N)^j/j!}\\
&=P_{par}(m)
\end{align*}
We want to find if Pareto service time satisfies the condition $r_{par}\ge r_{exp}$. From \eqref{Eq:time_pareto}, we have
    \begin{align*}
        \frac{\expec[T_j(m+1)]}{\expec[T_j(m)]}\ge\frac{m}{m+1} &\Leftrightarrow \frac{(m+1)(m\alpha-1)}{m(m\alpha+\alpha-1)}\ge\frac{m}{m+1}\\
        &\Leftrightarrow \alpha\ge \frac{2m+1}{m^2+m}
    \end{align*}
    The right side of the above formula decreases with increasing $m$. When $m\ge2$, $\frac{2m+1}{m^2+m}\le1$. For the Pareto distribution, the tail index $\alpha$ is larger than $1$. Therefore, the above condition always holds when $m\ge2$.
    When $m=1$, the formula shows that the condition holds when $\alpha\ge 1.5$. In contrast, when $\alpha< 1.5$, the job-blocking probability $P_b(m=2)$ may be smaller than $P_b(m=1)$.
\end{proof}
From the above lemma, we conclude that spreading computing resources over more jobs is generally a better way to support the system processing more jobs. However, the optimal replication factor depends on the tail parameter of the Pareto distribution.

\subsection{Job-blocking Probability versus Job-computing time}
From \eqref{Eq:time_shift} and \eqref{Eq:time_pareto}, we conclude that the expected job-computing time $\expec[T_{\text{job}}]$ increases with the number of groups $c$ and reaches its minimum at $c=1$ (i.e., $m=N$). However, from Theorem~\ref{Thm:block} and Lemma~\ref{Le:block_pareto}, we conclude that spreading computing resources over more jobs is generally a better way to support the system processing more jobs. Therefore, optimizing the job-blocking probability and the expected job-computing time with the variable $c$ is a dilemma.

We verify the above conclusions for both service time distributions in Figure~\ref{fig:timeprob}.
\begin{figure}[hbt]
    \centering
    \includegraphics[width=0.48\textwidth]{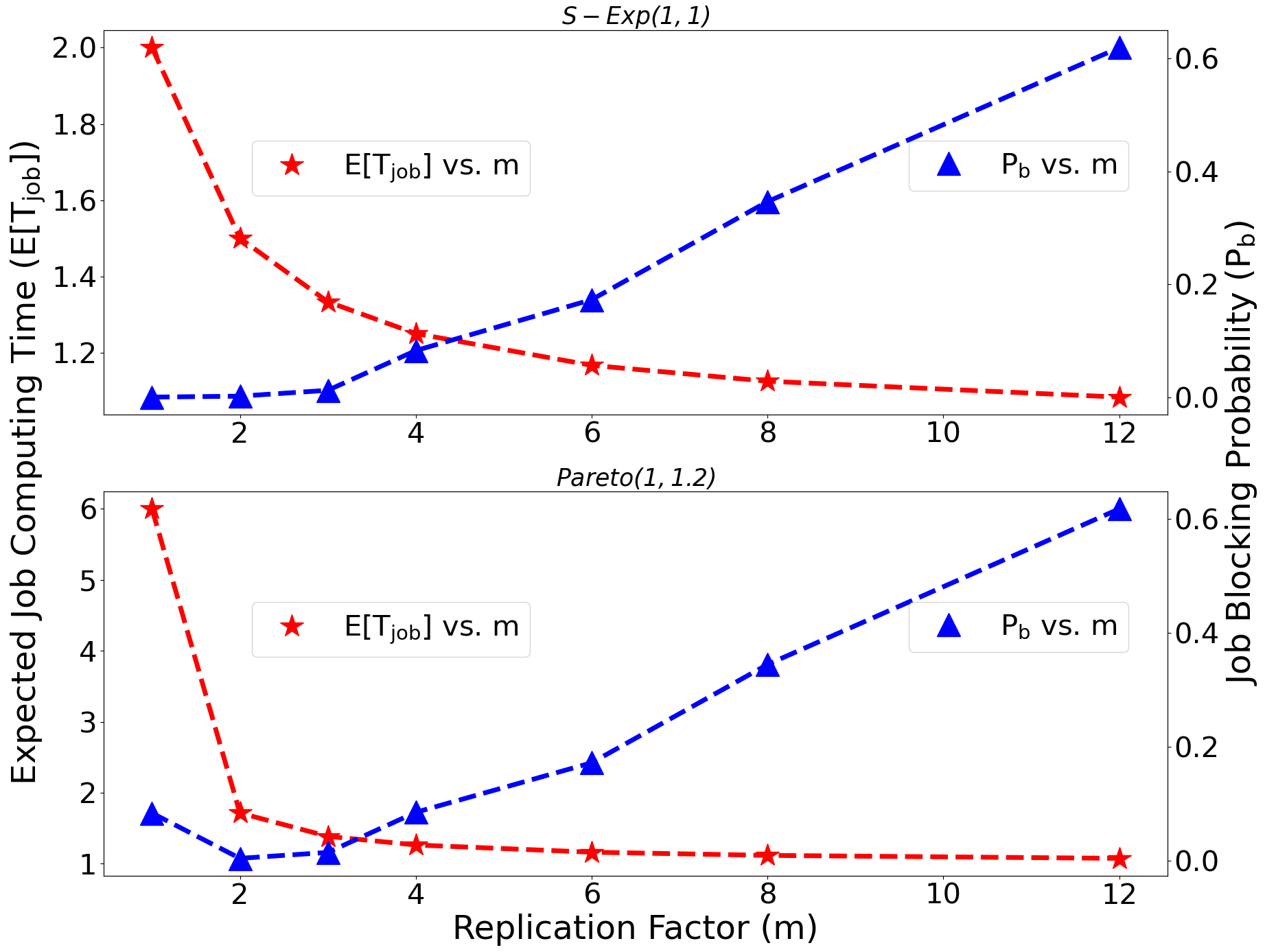}
    \caption{The expected job-computing time $\expec[T_{\text{job}}]$ and the job-blocking probability $P_b$ as a function of  $m$. The number of workers is $N=24$. The service time distribution are $\sexpD(1,1)$ (upper) and $\parD(1,1.2)$ (lower). }
    \label{fig:timeprob}
\end{figure}
The red star curves in both graphs show that introducing more redundancy increases computing speed. The blue triangle curves show that introducing redundancy results in more jobs sent to the cloud. For Pareto service time, the optimal replication factor is $m=2$.
The figure also suggests that some minimal replication significantly reduces job-computing time with almost no change in job-blocking probability.  

\subsection{Job-blocking Probability Optimization Algorithm}
The optimal replication factor for shifted exponential service time that minimizes the job-blocking probability is $m=1$ according to Theorem~\ref{Thm:block}. Thus, there is no reason to concentrate computing resources on a few jobs and introduce replication redundancy.
However, the optimal replication factor for Pareto service time that minimizes the job-blocking probability depends on the tail parameter $\alpha$ according to Lemma~\ref{Le:block_pareto}. Thus, we will propose Algorithm~\ref{Alg:1} in the following to verify the effectiveness of the redundant computing resource allocation.

\begin{algorithm}[hbt]
\caption{job-blocking probability Optimization (BPO)}
\label{Alg:1}
\begin{algorithmic}[1]
\Require Number of workers $N$, job arrival rate $\lambda$, Pareto service time $\parD(\alpha,\beta)$
\Ensure Optimal replication factor $m_o$
\If {$\alpha<1.5$}
    \State $P_1\gets \frac{(\lambda\alpha\beta/(\alpha-1))^N/N!}{\sum_{j=0}^{N}(\lambda\alpha\beta/(\alpha-1))^j/j!}$
    \State $k\gets \lfloor N/2 \rfloor$
    \State $P_2\gets \frac{(2\lambda \alpha\beta/(2\alpha-1))^k/k!}{\sum_{j=0}^{k}(2\lambda \alpha\beta/(2\alpha-1))^j/j!}$
    \If {$P_1\le P_2$}
        \State $m_o=1$
    \ElsIf {$P_1> P_2$}
        \State $m_o=2$
    \EndIf
\ElsIf{$\alpha\ge1.5$}
\State $m_o=1$
\EndIf
\end{algorithmic}
\end{algorithm}

\section{Average System Time }
\label{sec:ave_shift}

The previous section shows that when computing resources are concentrated on a few jobs, the edge system processes each job faster but blocks and sends more jobs to the cloud. Then, the edge system will be unstable even with a lower arrival rate. Therefore, a balance exists when optimizing the job-blocking probability and the job-computing time. Here, we analyze the average system time to evaluate the balance. We separately consider the shifted exponential and Pareto service times in the following and propose an average system time optimization (ASTO) algorithm to verify the effectiveness of the redundant resource allocation.

\subsection{Shifted Exponential Service Time}
According to the expression of the average system time in \eqref{Eq:avetime}, both the cloud time $\expec[T_{cl}]$ and the job-computing time $\expec[T_{job}]$ decide the performance of the $\expec[T_{sys}]$. Considering the complexity of the \eqref{Eq:avetime}, we can not find the optimal replication factor $m$ that minimizes the average system time. However, we observe that spreading computing resources to more jobs may be a better way to decrease the average system time when $\expec[T_{cl}]$ is large; Concentrating the resources on a few jobs may be better when $\expec[T_{job}]$ is large. In the following theorem, we find the optimal replication factor $m$ that minimizes the average system time $\expec[T_{\text{sys}}]$ under certain conditions.

\begin{theorem}
    \label{Thm:shift_ave}
    For the edge system with Poisson($\lambda$) arrivals and $\sexpD(\Delta,\mu)$ service time, when the cloud time satisfies $\expec[T_{\text{cl}}]\ge\frac{1}{\mu}\frac{1-\frac{1}{N}-P_b(N)+\frac{P_b(1)}{N}}{P_b(\lfloor N/2 \rfloor)-P_b(N)}+\Delta$, the average system time $\expec[T_{\text{sys}}]$ reaches the minimum at $m=1$ (i.e., $c=N$); when $\expec[T_{\text{cl}}]\in(\Delta+\frac{1}{N\mu},\Delta+\frac{1}{\mu})$, $\expec[T_{\text{sys}}]$ reaches the minimum at $m=N$ (i.e., $c=1$) under the condition $\lambda\le \frac{N\mu}{(N-2)(\Delta N\mu+1)}$.
\end{theorem}
\begin{proof}
When $c=N$ and $m=1$, the average system time is
    \[
    \expec[T_{\text{sys}}(N)]=(1-P_b(N))(\Delta+\frac{1}{\mu})+P_b(N)\expec[T_{cl}].
    \]
From \eqref{Eq:avetime} and \eqref{Eq:time_shift}, for any $c<N$ we have
    \begin{align*}
        &\expec[T_{\text{sys}}(N)]\le\expec[T_{\text{sys}}(c)]\\
         &\Leftrightarrow \expec[T_{\text{cl}}]\ge\frac{1}{\mu}\frac{1-\frac{1}{m}-P_b(N)+\frac{P_b(c)}{m}}{P_b(c)-P_b(N)}+\Delta.
    \end{align*}
    Since $P_b(c)$ decreases with the increasing $c$ according to Theorem~\ref{Thm:block}, the above inequality holds when
    \[\expec[T_{\text{cl}}]\ge\frac{1}{\mu}\frac{1-\frac{1}{N}-P_b(N)+\frac{P_b(1)}{N}}{P_b(\lfloor N/2 \rfloor)-P_b(N)}+\Delta.\]

    When $c=1$ and $m=N$, the average system time is
    \[\expec[T_{\text{sys}}(1)]=(1-P_b(1))(\Delta+\frac{1}{N\mu})+P_b(1)\expec[T_{\text{cl}}].\]
    Then, for any $c\ge2$ we have
    \begin{align*}
        &\expec[T_{\text{sys}}(1)]<\expec[T_{\text{sys}}(c)]\\
         &\Leftrightarrow \expec[T_{\text{cl}}]\le\frac{1}{\mu}\frac{\frac{1}{m}-\frac{1}{N}+\frac{P_b(1)}{N}-\frac{P_b(c)}{m}}{P_b(1)-P_b(c)}+\Delta.
         \end{align*}
    Since $\expec[T_{\text{cl}}]\le\Delta+\frac{1}{\mu}$, then the above inequality holds when
       \begin{align*}
        &\frac{\frac{1}{m}-\frac{1}{N}+\frac{P_b(1)}{N}-\frac{P_b(c)}{m}}{P_b(1)-P_b(c)}\ge1\\
         &\Leftrightarrow P_b(c)\ge\frac{(1-\frac{1}{N})P_b(1)-\frac{1}{m}+\frac{1}{N}}{1-\frac{1}{m}}.
    \end{align*}
    Similarly to the proof of Theorem~\ref{Thm:block}, the above inequality holds when $(1-\frac{1}{N})P_b(1)-\frac{1}{m}+\frac{1}{N}\le 0$. From the expression of $P_b(c)$ in \eqref{Eq:rewrite_shift}, we have
    \begin{align*}
         &(1-\frac{1}{N})\frac{K(\Delta N\mu+1)}{1+K(\Delta N\mu+1)}-\frac{1}{m}+\frac{1}{N}\le 0\\
         &\Leftrightarrow K(\Delta N\mu+1)\le\frac{1-\frac{m}{N}}{m-1}=\frac{1-\frac{1}{N}}{m-1}-\frac{1}{N}.
    \end{align*}
    Since the replication factor satisfies $m\le\frac{N}{2}$, we have $K\le\frac{1}{(N-2)(\Delta N\mu+1)}$, that is,  $\lambda\le\frac{N\mu}{(N-2)(\Delta N\mu+1)}$. 
\end{proof}
From the above theorem, we find some optimal values of $m$ or $c$ under certain conditions. Unlike the results for the job-blocking probability, Theorem~\ref{Thm:shift_ave} shows that spreading computing resources to more jobs is not always optimal, and the optimal replication factor changes with the system parameters such as $\lambda$, $N$, and $\expec[T_{cl}]$, etc. 

\subsection{Pareto Service Time}
For Pareto service time, we analyze how cloud time $\expec[T_{cl}]$ and job-computing time $\expec[T_{job}]$ affect the average system time. Similar to Theorem~\ref{Thm:shift_ave},  
we find the optimal replication factor $m$ that minimizes the average system time $\expec[T_{\text{sys}}]$ under certain conditions in the following lemma.
\begin{lemma}
    \label{Le:Pareto_ave}
    For the edge system with Poisson($\lambda$) arrivals and $\parD(\beta,\alpha)$ service time, if $P_b(N)\le P_b(\lfloor N/2 \rfloor)$, $\expec[T_{\text{sys}}]$ reaches the minimum at $m=1$ (i.e., $c=N$) when $\expec[T_{cl}]\ge \beta\alpha\frac{(1-P_b(N)/(\alpha-1))-(1-P_b(1))/(\alpha N-1)}{P_b(\lfloor N/2 \rfloor)-P_b(N)}$; otherwise, $\expec[T_{\text{sys}}]$ will not reach the minimum at $m=1$. 
    When $\expec[T_{\text{cl}}]\le \frac{\beta\alpha }{\alpha -1}$, $\expec[T_{\text{sys}}]$ reaches the minimum at $m=N$ (i.e., $c=1$) under the condition $\lambda\le\frac{(\alpha-1)N}{\beta\alpha N(N-2)}$.
\end{lemma}
\begin{proof}
    When $c=N$ and $m=1$, the average system time is
    \[
    \expec[T_{\text{sys}}(N)]=(1-P_b(N))\frac{\beta\alpha}{\alpha-1}+P_b(N)\expec[T_{cl}].
    \]
    From \eqref{Eq:avetime} and \eqref{Eq:time_pareto}, for any $c<N$ we have
    \begin{equation}
    \label{Eq:proof_pareto}
    \begin{aligned}
        &\expec[T_{\text{sys}}(N)]\le\expec[T_{\text{sys}}(c)]\\
        &\Leftrightarrow \expec[T_{cl}](P_b(c)-P_b(N))\\
        &\ge (1-P_b(N))\frac{\beta\alpha}{\alpha-1}-(1-P_b(c))\frac{\beta\alpha m}{\alpha m-1}.
    \end{aligned}
    \end{equation}
    According to Lemma~\ref{Le:block_pareto}, when $\alpha\ge 1.5$, $P_b$ increases with $m$. Then, the above inequality holds when
    \[ \expec[T_{cl}]\ge \beta\alpha\frac{(1-P_b(N)/(\alpha-1))-(1-P_b(1))/(\alpha N-1)}{P_b(\lfloor N/2 \rfloor)-P_b(N)} .\]
    When $\alpha< 1.5$, $P_b$ increases with $m$ for $m\ge2$. If $P_b(N)>P_b(\lfloor N/2 \rfloor)$, then $P_b(c)-P_b(N)<0$ for any $c<N$. For the job-blocking probability $P_b(c)$, it is easy to prove that $\frac{1}{2}\ge P_b(c)\ge0$ according to \eqref{Eq:blocking}.
    Thus, $\frac{1-P_b(c)}{1-P_b(N)}<2$. Since $\frac{\beta\alpha}{\alpha-1}< \frac{4\beta\alpha}{2\alpha-1}$ for $\alpha<1.5$, the right side of \eqref{Eq:proof_pareto} is always positive. Then $\expec[T_{\text{sys}}]$ will not reach the minimum at $c=N$.

When $c=1$ and $m=N$, the average system time is
    \[
    \expec[T_{\text{sys}}(1)]=(1-P_b(1))\frac{\beta\alpha m}{\alpha m-1}+P_b(N)\expec[T_{cl}].
    \]
    Then, for any $c\ge 2$ we have
    \begin{align*}
        &\expec[T_{\text{sys}}(1)]\le\expec[T_{\text{sys}}(c)]\\
        &\Leftrightarrow \expec[T_{cl}](P_b(1)-P_b(c))\\
        &\le (1-P_b(c))\frac{\beta\alpha m}{\alpha m-1}-(1-P_b(1))\frac{\beta\alpha N}{\alpha N-1}.
    \end{align*}
    Similar to the proof of Theorem~\ref{Thm:shift_ave}, since  $\expec[T_{\text{cl}}]\le\frac{\beta\alpha }{\alpha -1}$, then the above inequality holds when
    \begin{align*}
        &\frac{m(1-P_b(c))}{\alpha m -1}-\frac{N(1-P_b(1))}{\alpha N-1}\ge \frac{P_b(1)-P_b(c)}{\alpha-1}\\
        & P_b(c)\frac{m-1}{\alpha m -1}\ge \frac{(N-1)P_b(1)}{\alpha N-1}-\frac{(\alpha-1)(N-m)}{(\alpha N-1)(\alpha m -1)}
    \end{align*}
    Since $P_b(c)\ge 0$, the above inequality holds when the right side is negative. Then we have
    \[\lambda\le\frac{(\alpha-1)(N-m)}{\beta\alpha N(m-1)}.\]
    Since $m\le\frac{N}{2}$, $\expec[T_{\text{sys}}]$ reaches the minimum at $c=1$ when $\lambda\le\frac{(\alpha-1)N}{\beta\alpha N(N-2)}$.
\end{proof}
Similarly to the conclusion of Theorem~\ref{Thm:shift_ave}, the optimal replication factor $m$ minimizes the average system time changes with different values of $\expec[T_{\text{cl}}]$ and $\lambda$. Besides, the tail parameter $\alpha$ of the Pareto distribution is also an essential factor that decides the optimal replication factor.

\subsection{Average System Time Optimization}

Here, we propose the average system time optimization algorithm according to the analysis in Theorem~\ref{Thm:shift_ave} and Lemma~\ref{Le:Pareto_ave}. Notice that the previous study is based on the assumption that the edge system is stable. However, when the values of $m$ and $c$ change, the system stability condition may change accordingly. Then, the system may be unstable due to a large arrival rate.

For the blocking system, the system is stable when $\rho=\lambda\expec[T_{\text{job}}]/c\le1$. From \eqref{Eq:time_shift} and \eqref{Eq:time_pareto}, we find the system stability conditions for shifted exponential and Pareto service time separately. 
\[\rho_e(c)=\lambda(\frac{\Delta}{c}+\frac{1}{N\mu}) 
\; \;\text{and}\; \;\rho_p(c)=\frac{\lambda\beta\alpha N}{\alpha N c-c^2}.\]
From the above equations, we know that the maximum arrival rate $\lambda_m$ under which the system is stable changes with the number of groups $c$. Then we have
\begin{equation*}
\lambda_m(c) = \begin{cases}
\frac{\mu Nc}{\Delta\mu N+c}  & \text{for} \quad  \sexpD(\Delta,\lambda) \\
\frac{\alpha N c-c^2}{\beta\alpha N} & \text{for} \quad \parD(\beta,\alpha)
\end{cases} \label{eqn:toy_X_dist}
\end{equation*}
When the arrival rate $\lambda\le \lambda_m$, the results in Theorem~\ref{Thm:shift_ave} and Lemma~\ref{Le:Pareto_ave} can be applied to the algorithm directly. When $\lambda> \lambda_m$, the average system time from \eqref{Eq:avetime} does not hold since the blocking system is unstable. Therefore, we modify the expression of the average system time under this condition.
\begin{equation}
    \label{Eq:aver_modify}
    \begin{aligned}
        \expec[T_{\text{sys}}(c)]&=\frac{\lambda_m(c)(1-P_{max}(c))}{\lambda}\expec[T_{\text{job}}(c)]\\&+(1-\frac{\lambda_m(c)}{\lambda}(1-P_{max}(c)))\expec[T_{\text{cl}}].
    \end{aligned}
\end{equation}
Where $P_{max}(c)$ is defined as the maximum job-blocking probability and the formula is 
\begin{equation}
    \label{Eq:blocking_max}
    P_{max}(c)=\frac{(c)^c/c!}{\sum_{j=0}^{c}(c)^j/j!}.
\end{equation}
Finally, we propose the ASTO algorithm in the following.

\begin{algorithm}[hbt]
\caption{Average System Time Optimization (ASTO)}
\label{Alg:2}
\begin{algorithmic}[1]
\Require Number of workers $N$, job arrival rate $\lambda$, cloud time $T_{cl}$, $\sexpD(\Delta,\mu)$ service time
\Ensure Optimal replication factor $m_o$
\State $c_1\gets N$, $c_2\gets 1$ and $c_3=\lfloor N/2 \rfloor$
\State Calculate $P_1$, $P_2$ and $P_3$ for each $c_i$ ($i\in\{1,2,3\}$) according to \eqref{Eq:rewrite_shift} 
\State $T_1\gets\frac{1}{\mu}\frac{1-\frac{1}{N}-P_1+\frac{P_2}{N}}{P_3-P_1}+\Delta$ 
\State $\lambda_1\gets \frac{N\mu}{(N-2)(\Delta N\mu+1)}$
\If{$T_{cl}\ge T_1$}
    \State $m_o=1$
\ElsIf{$T_{cl}\le \Delta+\frac{1}{\mu}$ and $\lambda\le\lambda_1$}
    \State $m_o=N$
\Else
\For {$m \gets 1$ to $N$}
    \State $c\gets \lfloor N/m \rfloor$
    \State $\lambda_m\gets\frac{\mu Nc}{\Delta\mu N+c}$
    \State $T_j\gets \Delta+1/m\mu$ 
    \If{$\lambda\le\lambda_m$}
        \State $P\gets \frac{(\lambda T_j)^c/c!}{\sum_{j=0}^{c}(\lambda T_j)^j/j!}$
        \State $T_{sys}(m)\gets (1-P)T_j+PT_{cl}$
    \ElsIf{$\lambda>\lambda_m$}
        \State $P\gets \frac{(c)^c/c!}{\sum_{j=0}^{c}(c)^j/j!}$
        \State $T_{sys}(m)\gets \frac{\lambda_m }{\lambda}(1-P)T_j+(1-\frac{\lambda_m }{\lambda}+\frac{\lambda_m }{\lambda}P)T_{cl}$
    \EndIf
\EndFor
\State $M\gets {1,2,...,N}$
\State $m_o=\arg\min\limits_{v\in M} T_{sys}(v)$
\EndIf
\end{algorithmic}
\end{algorithm}

\section{Average System Time Analysis for Different Parameters}
\label{Sec:para}

\subsection{Cloud Time and Job Arrival Rate}
The edge system can provide better computing performance by reducing the communication time between the users and the cloud. If $\expec[T_{\text{cl}}]\le \expec[T_{\text{job}}]$, we consider that the edge system does not have enough computing resources for each job and should always send all the jobs to the cloud. Therefore, when $\expec[T_{\text{cl}}]\in(\expec[(T_{\text{job}}(m=N)],\expec[T_{\text{job}}(m=1)])$, the optimal replication factor $m$ should satisfy $\expec[T_{\text{job}}(m)]<\expec[T_{\text{cl}}]$. It is always better to concentrate computing resources on a few jobs.

When cloud time satisfies $\expec[T_{\text{cl}}]>\expec[T_{\text{job}}(m=1)]$, it is always faster to process the job in the edge system. From \eqref{Eq:blocking_max}, when the edge system is stable, it is easy to prove 
\[P_b(c)\le P_{max}(c)=\frac{1}{\sum_{j=0}^{c}(c!/j!)/(c)^{c-j}}\le\frac{1}{2}.\]
Thus, we have $1-P_b(c)\ge P_b(c)$, which means more jobs will be processed in the edge system.
From the expression of average system time \eqref{Eq:avetime}, when $\expec[T_{\text{cl}}]$ is relatively small, minimizing the job-computing time may reduce the average system time. Thus, concentrating computing resources on a few jobs may improve performance. When $\expec[T_{\text{cl}}]$ is large, reducing the blocked jobs will reduce the average system time. Thus, spreading computing resources to more jobs may improve performance. According to Theorem~\ref{Thm:shift_ave}, when $\expec[T_{\text{cl}}]$ is sufficiently large, maximum spreading computing resources ($m=1$) is optimal for shifted exponential service time. However, according to Lemma~\ref{Le:Pareto_ave}, for Pareto service time, when the tail parameter $\alpha$ is small, the straggler problem in the edge system is severe, so $m=2$ is optimal to minimize the average system time.

The job arrival rate $\lambda$ is also an essential factor to affect the average system time. By \eqref{Eq:blocking}, when $\lambda$ is small, the job-blocking probability is close to $0$. Then, the edge system will send fewer jobs to the cloud. From \eqref{Eq:avetime}, we see that minimizing $\expec[T_{\text{sys}}]$ is equivalent to minimize $(1-P_b(c))\expec[T_{\text{job}}]$. Thus, concentrating computing resources on a few jobs is a better strategy. Furthermore, by Theorem~\ref{Thm:shift_ave} and Lemma~\ref{Le:Pareto_ave}, when the cloud time $\expec[T_{\text{cl}}]<\expec[T_{\text{job}}(m=1)]$, maximum concentrating computing resources ($m=N$) is optimal. 
When $\lambda$ is large, more jobs will be sent to the cloud. Since $1-P_b(c)\ge P_b(c)$, the optimal $m$ that minimizes the average system time will depend on the value of $\expec[T_{\text{cl}}]$.

\subsection{Number of Workers in Edge System}
\label{Sub:largeN}
The number of workers $N$ is one of the most critical parameters that decides the system's performance. As the increase of $N$, since $N=mc$, when the replication factor $m$ increases with $N$, the job-computing time will decrease with the increasing $N$ according to \eqref{Eq:time_shift} and \eqref{Eq:time_pareto}; when the number of groups $c$ increases with $N$, the job-blocking probability will decrease with increasing $N$ according to \eqref{Eq:blocking}. 
Therefore, it is necessary to study how the $m$ and $c$ increase with $N$ so that the average system time can reach the minimum. Here, we consider three possible scenarios:

\begin{enumerate}[label=\noindent Scenario \arabic*:,leftmargin=*]
\item $m$ increases with $N$, and $c$ remains the same.

\item  $c$ increases with $N$, and $m$ remains the same.

\item both $m$ and $c$ increase with $N$.
\end{enumerate}
 We will analyze the three scenarios under the condition that the number of workers $N$ goes to infinity. The scenarios change as follows:
In \noindent Scenario 1, a constant $c_1$ exists that $c$ is always smaller than $c_1$. 
In \noindent Scenario 2, a constant $m_1$ exists that $m$ is always smaller than $m_1$. In \noindent Scenario 3, the replication factor $m$ and the number of groups $c$ go to infinity.
First, we consider the shifted exponential service time. We find the minimum average system time and optimal replication factor under three scenarios in the following theorem.

\begin{theorem}
\label{Thm:sce}
    In the edge system with $\sexpD(\Delta,\mu)$ service time, as the number of workers $N$ goes to infinity, the average system time reaches the minimum at $\Delta$ when $m$ and $c$ go to infinity accordingly.
\end{theorem}
\begin{proof}
      We adopt the expression of the job-blocking probability from \eqref{Eq:time_shift} and \eqref{Eq:blocking}, then
    \[P_{b}=\frac{(\lambda \Delta+\frac{\lambda }{m\mu})^{c}/c!}{\sum_{j=0}^{c}(\lambda \Delta+\frac{\lambda }{m\mu})^j/j!}.\]
    By \eqref{Eq:avetime} and \eqref{Eq:time_shift}, the average system time is
    \begin{align*}
        \expec[T_{\text{sys}}]= (\Delta+\frac{1}{m\mu})(1-P_{b})+\expec[T_{\text{cl}}]P_{b}.
    \end{align*}
\noindent Scenario 1: 
Since we assume that $c$ is finite, we have 
\begin{align*}
    \lim\limits_{m \to \infty}P_{b}=\frac{ \lim\limits_{m \to \infty}(\lambda \Delta+\frac{\lambda }{m\mu})^{c}/c!}{\lim\limits_{m \to \infty}\sum_{j=0}^{c}(\lambda \Delta+\frac{\lambda }{m\mu})^j/j!}
    =\frac{(\lambda\Delta)^c/c!}{\sum_{j=0}^{c}(\lambda\Delta)^j/j!}.
\end{align*}
Since $\expec[T_{\text{cl}}]$ is a constant and $c\le c_1$,
\begin{align*}
     \lim\limits_{m \to \infty} \expec[T_{\text{sys}}]=\Delta+\frac{(\expec[T_{\text{cl}}]-\Delta)(\lambda\Delta)^{c_1}/c_1!}{\sum_{j=0}^{c_1}(\lambda\Delta)^j/j!}.
\end{align*}   

\noindent Scenario 2: 
Since $m$ and $\Delta$ are finite, $\lambda\Delta+\frac{\lambda}{m\mu}$ is also finite. Then, we have 
\[\lim\limits_{c\to \infty}(\lambda \Delta+\frac{\lambda }{m\mu})^{c}/c!=0\]
and
\[\lim\limits_{c \to \infty}\sum_{j=0}^{c}(\lambda\Delta+\frac{\lambda }{m\mu})^j/j!=e^{\lambda\Delta+\frac{\lambda}{m\mu}}.\]
Thus, we have $\lim\limits_{c \to \infty}P_{b}=0$. Since $\expec[T_{\text{cl}}]$ is a constant and $m\le m_1$,
\begin{align*}
     \lim\limits_{c \to \infty} \expec[T_{\text{sys}}]=\Delta+\frac{1}{m_1\mu}.
\end{align*}

\noindent Scenario 3:
The job-blocking probability satisfies $P_b\ge0$; thus  
\begin{align*}
    &P_{b}\le\frac{(\lambda \Delta+\frac{\lambda }{m\mu})^{c}/c!}{1+(\lambda \Delta+\frac{\lambda }{m\mu})^c/c!}\\
    &= \frac{1}{1+c!/(\lambda \Delta+\frac{\lambda }{m\mu})^{c}}
    \le \frac{1}{1+c!/(\lambda \Delta+\frac{\lambda }{\mu})^{c}}.
\end{align*}
As $c$ goes to infinity, $c!$ increases much faster than $ (\lambda \Delta+\frac{\lambda }{\mu})^{c}$. Thus, we have
\[ \lim\limits_{\substack{c \to \infty \\ m \to \infty}}\frac{1}{1+c!/(\lambda \Delta+\frac{\lambda }{\mu})^{c}}=0.\]
Thus, we have $\lim\limits_{\substack{c \to \infty \\ m \to \infty}}P_b=0$.
Finally, we get
\[\lim\limits_{\substack{c \to \infty \\ m \to \infty}} \expec[T_{\text{sys}}]=\lim\limits_{ m \to \infty} (\Delta +\frac{1}{m\mu})=\Delta.\]

\end{proof}

In the proof of Theorem~\ref{Thm:sce}, we can compare the minimum values of $\expec[T_{\text{sys}}]$, it is evident that the minimum average system time with  Scenario 3 reaches the overall minimum at $\Delta$. Therefore, to get the best system performance, we should let the replication factor $m$ and the number of groups $c$ increase with $N$. When considering the exponential service time (special case when $\Delta=0$), it is not hard to see  Scenario 1 performing similarly to Scenario 3. Then, we only need to let the replication factor $m$ increase with $N$ to achieve the minimum average system time.

Next, we consider the Pareto service time. Similarly, in the following lemma, we find the minimum average system time and optimal replication factor under three scenarios.

\begin{lemma}
\label{Lm:sce_pareto}
    In the edge system with $\sexpD(\Delta,\mu)$ service time, as the number of workers $N$ goes to infinity, the average system time reaches the minimum at $\beta$ when $m$ and $c$ go to infinity accordingly.
\end{lemma}
\begin{proof}
    As in the proof of Theorem~\ref{Thm:sce}, we adopt the expression of the job-blocking probability from \eqref{Eq:time_pareto} and \eqref{Eq:blocking}. Then
    \[P_{b}=\frac{(\frac{\lambda\beta\alpha m}{\alpha m-1})^{c}/c!}{\sum_{j=0}^{c}(\frac{\lambda\beta\alpha m}{\alpha m-1})^j/j!}.\]
    By \eqref{Eq:avetime} and \eqref{Eq:time_pareto}, the average system time is
    \begin{align*}
\expec[T_{\text{sys}}]= (\frac{\beta\alpha m}{\alpha m-1})(1-P_{b})+\expec[T_{\text{cl}}]P_{b}.
    \end{align*}
\noindent Scenario 1: 
Since we assume that $c$ is finite, we have 
\begin{align*}
    \lim\limits_{m \to \infty}P_{b}=\frac{ \lim\limits_{m \to \infty}(\frac{\lambda\beta\alpha m}{\alpha m-1})^{c}/c!}{\lim\limits_{m \to \infty}\sum_{j=0}^{c}(\frac{\lambda\beta\alpha m}{\alpha m-1})^j/j!}
    =\frac{(\lambda\beta)^c/c!}{\sum_{j=0}^{c}(\lambda\beta)^j/j!}.
\end{align*}
Since $\expec[T_{\text{cl}}]$ is a constant and $c\le c_1$,
\begin{align*}
     \lim\limits_{N \to \infty} \expec[T_{\text{sys}}]=\beta+\frac{(\expec[T_{\text{cl}}]-\beta)(\lambda\beta)^{c_1}/c_1!}{\sum_{j=0}^{c_1}(\lambda\beta)^j/j!}.
\end{align*}   

\noindent Scenario 2: 
Similarly to the proof in Theorem~\ref{Thm:sce}, we have $\lim\limits_{c\to \infty}(\frac{\lambda\beta\alpha m}{\alpha m-1})^{c}/c!=0$ and $\lim\limits_{c \to \infty}\sum_{j=0}^{c}(\frac{\lambda\beta\alpha m}{\alpha m-1})^j/j!=e^{\frac{\lambda\beta\alpha m}{\alpha m-1}}$. Thus, the job-blocking probability satisfies $\lim\limits_{c \to \infty}P_{b}=0$. Since $m\le m_1$
\begin{align*}
     \lim\limits_{c \to \infty} \expec[T_{\text{sys}}]=\frac{\beta\alpha m_1}{\alpha m_1-1}.
\end{align*}

\noindent Scenario 3:
The job-blocking probability satisfies $P_b\ge0$, and
\begin{align*}
    P_{b}\le \frac{1}{1+c!/(\frac{\lambda\beta\alpha m}{\alpha m-1})^{c}}
    \le \frac{1}{1+c!/(\frac{\lambda\beta\alpha }{\alpha -1})^{c}}.
\end{align*}
As $c$ goes to infinity, 
 $\lim\limits_{\substack{c \to \infty \\ m \to \infty}}\frac{1}{1+c!/(\frac{\lambda\beta\alpha }{\alpha -1})^{c}}=0$.
Thus, we have $\lim\limits_{\substack{c \to \infty \\ m \to \infty}}P_{b}=0$. Therefore, 
\[\lim\limits_{\substack{c \to \infty \\ m \to \infty}} \expec[T_{\text{sys}}]=\lim\limits_{ m \to \infty} \frac{\beta\alpha m }{\alpha m -1}=\beta.\]
\end{proof}
From Lemma~\ref{Lm:sce_pareto}, we obtain similar results to Theorem~\ref{Thm:sce}, that is, the minimum average system time with \noindent Scenario 3 reaches the overall minimum at $\beta$. Therefore, to get the best system performance, we should let the replication factor $m$ and the number of groups $c$ increase with $N$.

\section{Simulation and Numerical Results}

In this section, we first verify the effectiveness of redundancy management in improving system performance. We evaluate the performance of the BPO and ASTO algorithms for the shifted exponential and Pareto service times separately. The benchmark does not apply any redundant computing resource allocation strategy, so the replication factor is $m=1$, and the number of groups is $c=N$. 

In Figure~\ref{fig:sim}, we evaluate the average system time $\expec[T_{sys}]$ versus the job arrival rate $\lambda$ for different algorithms. We consider $N=24$ workers in the edge system and the cloud time $T_{cl}=15$. In the upper graph, the system uses $\sexpD(3,0.2)$ service time; in the lower graph, the system uses $\parD(1.5,1.2)$ service time. We consider $500$ arrivals and calculate the average system time for each job. Then, we simulate the system $100$ times to get the average system time.
\begin{figure}[htb]
    \centering
    \includegraphics[width=0.48\textwidth]{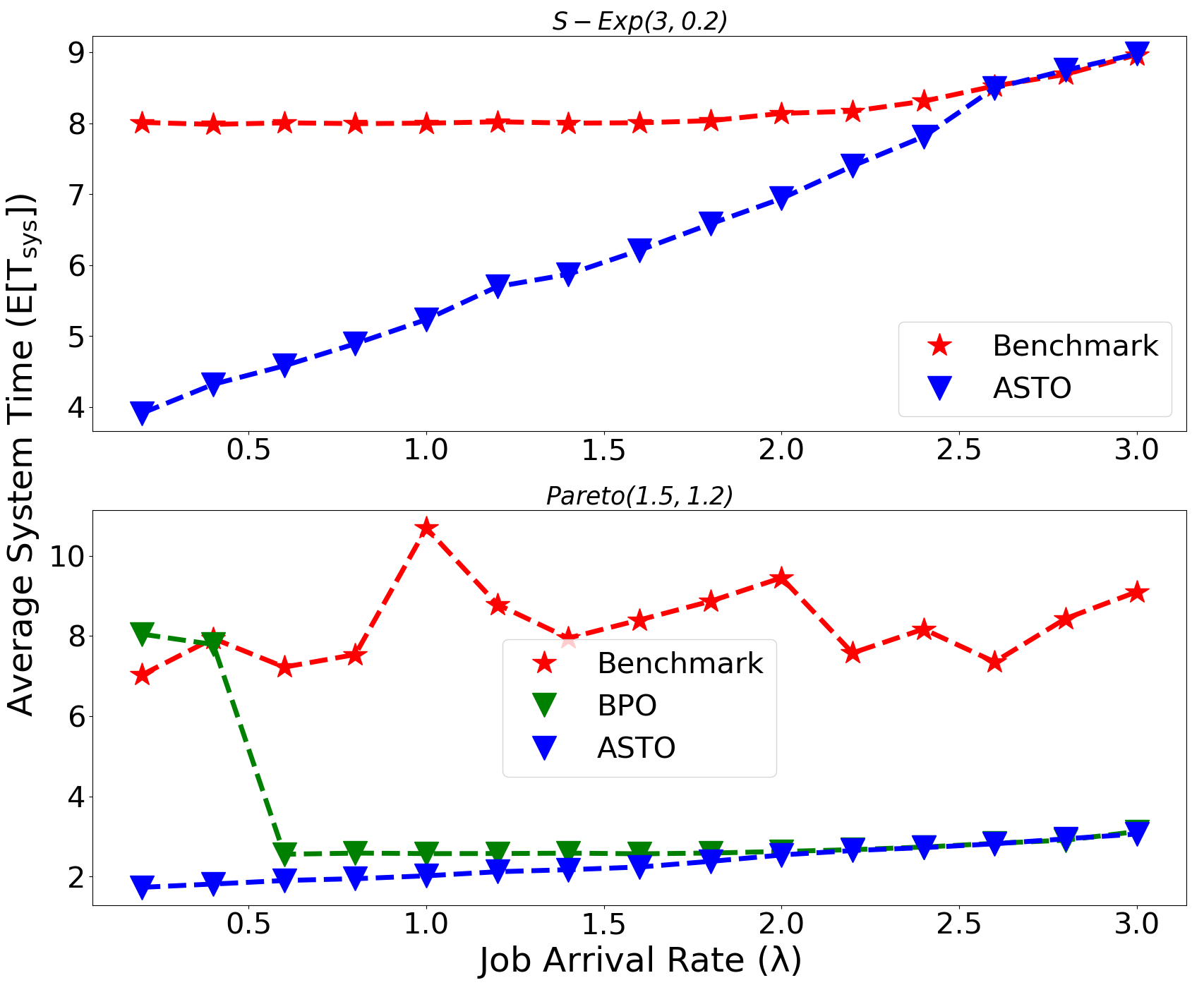}
    \caption{ The average system time $\expec[T_{sys}]$ versus the arrival rate $\lambda$ for different algorithms. The parameter settings are $N=24$ and $T_{cl}=15$. The service time distribution are $\sexpD(3,0.2)$ (upper) and $\parD(1.5,1.2)$ (lower).}
    \label{fig:sim}
\end{figure}

The upper graph only compares the ASTO algorithm with the benchmark since the BPO algorithm performs the same as the benchmark under shifted exponential service time. When $\lambda\le 2.5$, the benchmark almost does not change with the arrival rate, and the ASTO algorithm significantly outperforms the benchmark, but the improvement decreases with the arrival rate. When $\lambda>2.5$, both algorithms perform the same and increase with the arrival rate. We conclude that the ASTO algorithm can significantly improve system performance when the arrival rate is relatively small. 

The lower graph compares the ASTO and BPO algorithms with the benchmark.
When $\lambda\le 2$, the ASTO algorithm can reduce the average system time by about $65\%$ compared to the benchmark and by about $15\%$ compared to the BPO algorithm. When $\lambda>2$, ASTO and BPO algorithms perform the same and are better than the benchmark. When $\lambda\le 0.5$, the BPO algorithm performs similarly to the benchmark. 
We conclude that the ASTO algorithm significantly outperforms the other two algorithms. The BPO algorithm is better than the benchmark for relatively large job arrival rate scenarios.

Next, we numerically analyze the different parameters that affect the system performance. We consider $N=24$ workers in the edge system and evaluate the expressions of the average system time for $\sexpD(3,0.2)$ service time and $\parD(1.5,1.2)$ service time separately. 

In Figure~\ref{fig:time_lambda}, we evaluate \eqref{Eq:avetime} and \eqref{Eq:aver_modify} to see how the average system time changes with the replication factor. We consider that the cloud time $\expec[T_{\text{cl}}]=8$ is smaller than the maximum job-computing time for both service time distributions and plot  $\expec[T_{\text{sys}}]$ versus $m$ for different values of $\lambda\in \{0.1,1,5\}$. 
\begin{figure}[hbt]
	\centering
	\includegraphics[width=0.48\textwidth]{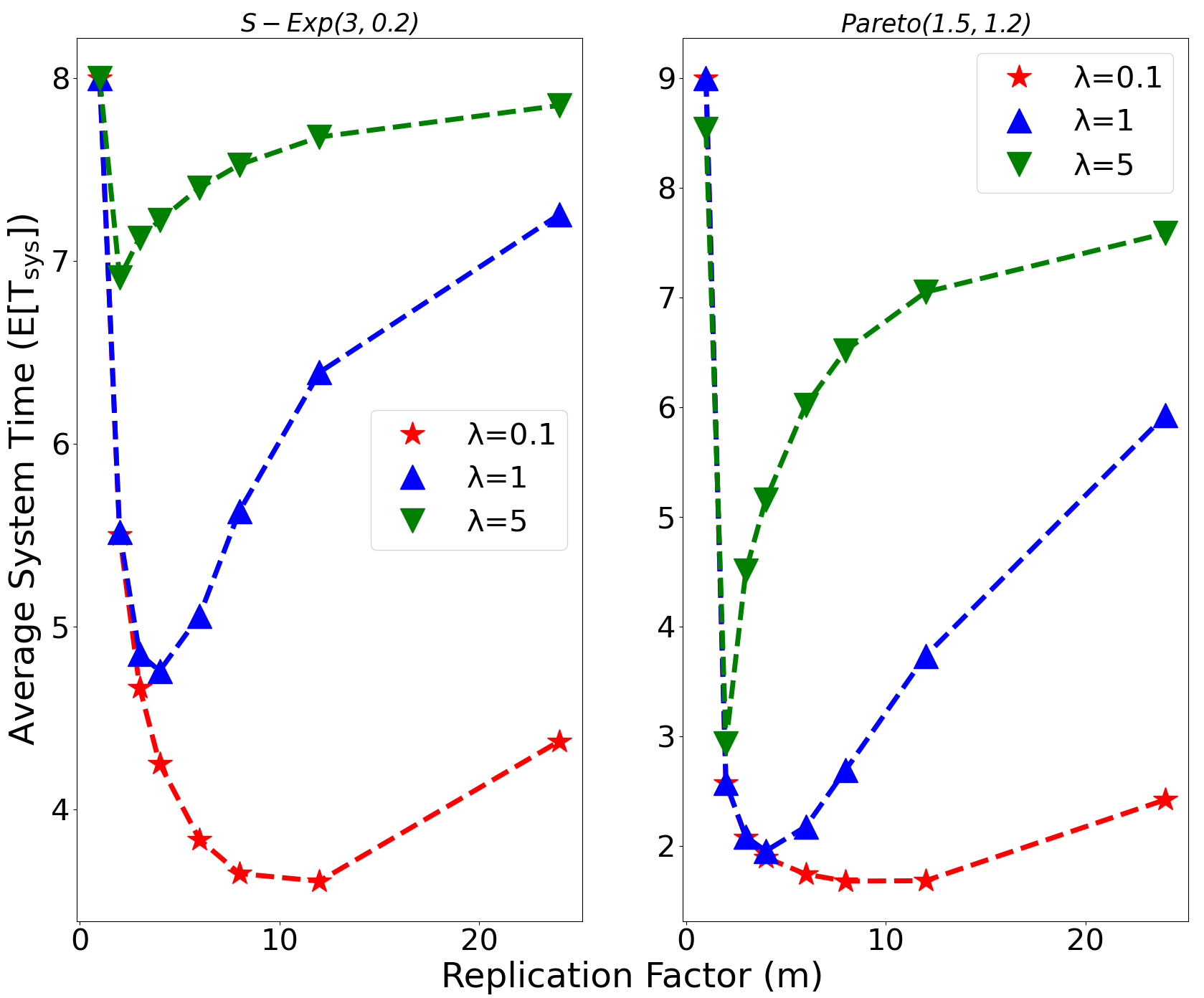}~
	\caption{ The average system time $\expec[T_{sys}]$ versus the arrival rate $\lambda$ for different values of $\lambda\in \{0.1,1,5\}$. The parameter settings are $N=24$ and $T_{cl}=8$. The service time distribution are $\sexpD(3,0.2)$ (left) and $\parD(1.5,1.2)$ (right). Increasing the replication factor properly leads to a smaller average system time when $\lambda$ is small.}
	\label{fig:time_lambda}
\end{figure}
The figure shows $\expec[T_{\text{sys}}]$ decreases with the increasing $m$, and then increases with $m$. So there exists an optimal value of $m$ between $1$ and $N$. Compared to different values of $\lambda$, the optimal $m$ decreases with the increasing $\lambda$, and the minimum value of $\expec[T_{\text{sys}}]$ increases significantly with $\lambda$. From the results, we have the following conclusions. Finding a proper resource allocation strategy can significantly reduce the average system time. With the increase in the arrival rate, the strategy becomes ineffective because the workers in the edge system are always busy in this situation.

In Figure~\ref{fig:time_cloud}, we plot the average system time  $\expec[T_{\text{sys}}]$ versus the replication factor $m$ for different values of $\expec[T_{\text{cl}}]\in\{15,50,100\}$. We consider the arrival rate $\lambda=1.5$.
\begin{figure}[hbt]
	\centering
	\includegraphics[width=0.48\textwidth]{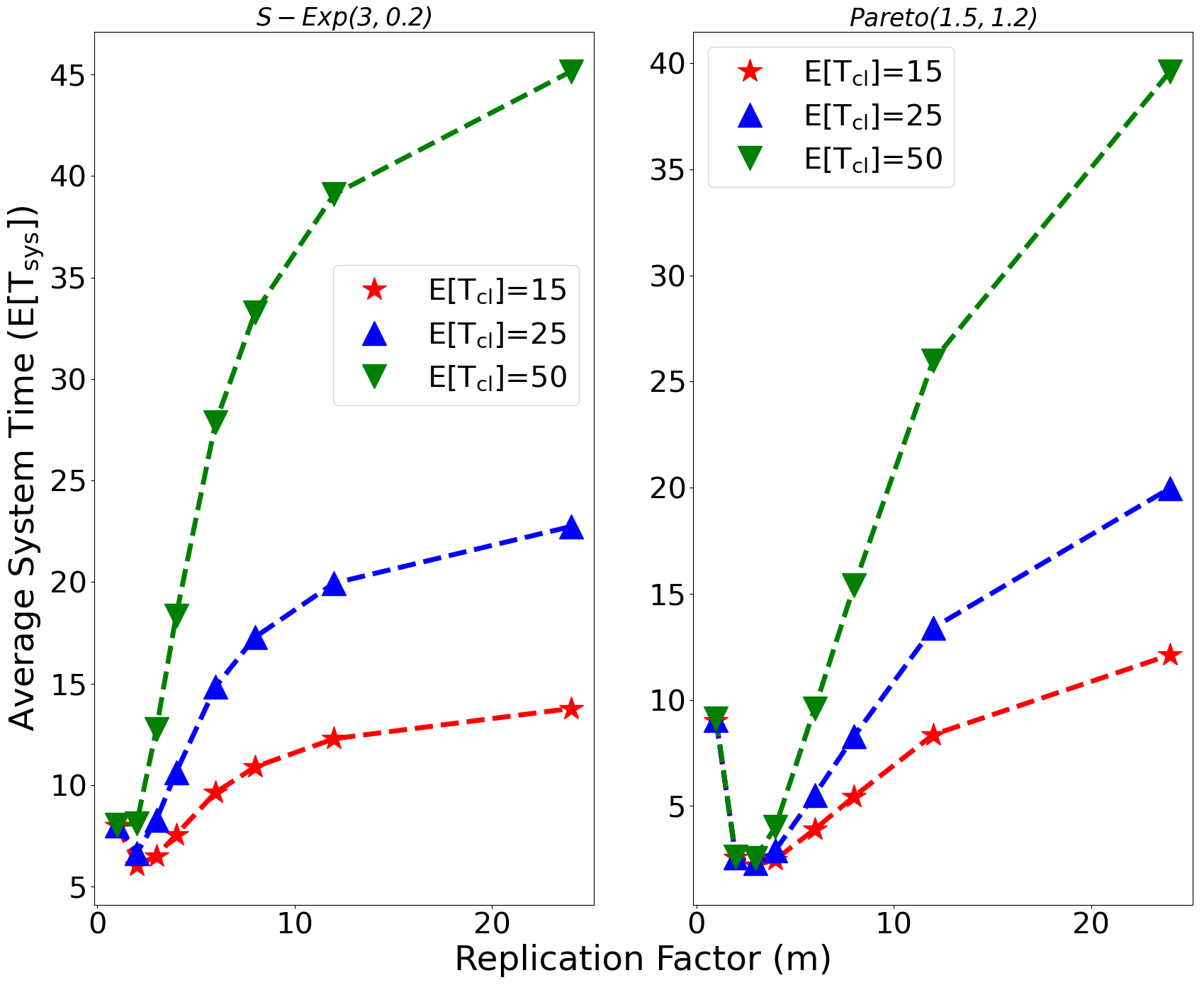}~
	\caption{ The average system time $\expec[T_{sys}]$ versus the arrival rate $\lambda$ for different values of $\expec[T_{\text{cl}}]\in \{15,50,10\}$. The parameter settings are $N=24$ and $\lambda=1.5$. The service time distribution are $\sexpD(3,0.2)$ (left) and $\parD(1.5,1.2)$ (right). When $\expec[T_{\text{cl}}]$ is small, introducing proper redundancy can decrease the average system time.}
	\label{fig:time_cloud}
\end{figure}
The figure shows that the optimal value of $m$ increases with the increasing $\expec[T_{\text{cl}}]$. For shifted exponential service time, $m=1$ is optimal when $\expec[T_{\text{cl}}]=50$; But, for Pareto service time, $m=1$ is not optimal even when $\expec[T_{\text{cl}}]$ is large. The observations are consistent with the conclusions in Theorem~\ref{Thm:shift_ave} and Lemma~\ref{Le:Pareto_ave}.
When $\expec[T_{\text{cl}}]$ is large enough, minimizing $\expec[T_{\text{sys}}]$ is equivalent to minimize the job-blocking probability $P_b$. 

Finally, we numerically analyze the optimal replication factor $m$ and the number of groups $c$ change with the number of workers $N$. 
In Figure~\ref{fig:Opt_c_m}, we evaluate \eqref{Eq:time_shift} and \eqref{Eq:avetime} to find the optimal $c$ and $m$ that minimize the average system time. Then we plot the optimal $c$ versus $N$ and the optimal $m$ versus $N$ for both service times $\sexpD(3,0.2)$ and $\parD(1.5,1.2)$. 
We consider $N=24$ workers, the arrival rate $\lambda=1.5$, the cloud time $\expec[T_{\text{cl}}]=10$. Since $m$ and $c$ are integers, we assume $c=\lfloor \frac{N}{m}\rfloor$. Although the actual values of $\expec[T_{\text{sys}}]$ may become small, they do not affect the conclusions of the numerical analysis.
\begin{figure}[hbt]
    \centering
    \includegraphics[width=0.48\textwidth]{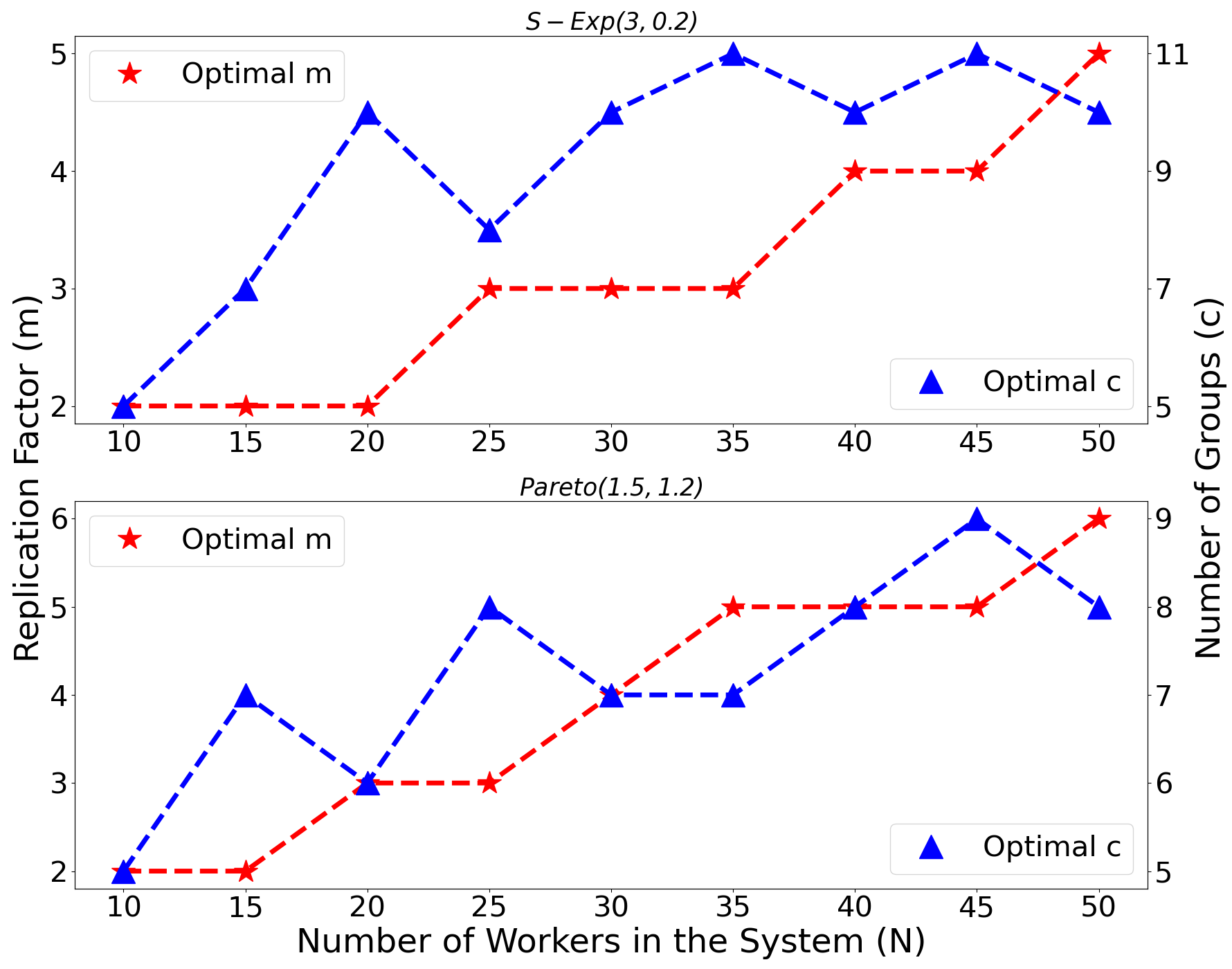}
    \caption{The number of groups $c$ and the repliction factor $m$ versus the number of workers $N$. The parameter settings are $\lambda=1.5$ and $\expec[T_{\text{cl}}]=10$. The service time distribution are $\sexpD(3,0.2)$ (upper) and $\parD(1.5,1.2)$ (lower). Since $m$ and $c$ are integers, we consider $m=\lfloor \frac{N}{c}\rfloor$. Both the optimal $m$ and $c$ increases with $N$.}
    \label{fig:Opt_c_m}
\end{figure}
We observe that the optimal $m$ and the optimal $c$ increase with $N$. The optimal $m$ increases stepwise, and the optimal $c$ increases fast when $N$ is small and increases smoothly when $N$ is large.
These observations are consistent with the conclusions in Section~\ref{Sub:largeN}.
Notice that the fluctuations on the curve of the optimal $c$ are caused by the fact that $c$ can not take every integer value. After all, $N$, $m$, and $c$ are all integers with the relation $N=mc$.

\section{Conclusion and Future Directions}
\label{sec:conclusions}

We consider an edge computing system with limited storage and computing resources in which the arrival jobs will be sent to the cloud when the system is busy. First, we addressed the problems concerning the number of groups that optimize the job-computing time and the job-blocking probability. We find it impossible to simultaneously minimize the job-computing time and the job-blocking probability for both shifted exponential and Pareto service times. Second, we use the average system time to evaluate the trade-off between the above two metrics. 
We find different optimal replication factors under certain conditions that minimize the average system time and propose the ASTO algorithm to verify the effectiveness of redundant computing resource allocation. 
Third, we study how the average system time and the optimal replication factor change with the job arrival rate, the cloud time, and the number of workers, and summarize the results in Table~\ref{Table:ave}.
Finally, the simulation results show that the ASTO algorithm significantly outperforms the benchmark for both shifted exponential and Pareto service times. The optimal replication factor changes with different system parameters, and the numerical analysis is consistent with the previous theoretical analysis.

This work sets the stage for many problems of interest to be studied in the future. We briefly describe three directions of immediate interest.
\subsubsection{Job-computing Time and job-blocking probability trade-off for M/G/c/k Queues}

As we mentioned in Section~\ref{subsec:formula}, our edge computing system is generally modeled as an M/G/c/k queue, where the queue capacity $k>c$. When all servers are busy, the new arrival job will wait in the limited-length queue, and the job will be blocked when the queue gets fully occupied. When we concentrate the computing resources on a few jobs to obtain a smaller computing time, the queue's capacity does not decrease linearly. Thus, compared to the M/G/c/c queue results, spreading resources may not be the best strategy considering the job-blocking probability. 
\subsubsection{Job-computing Time and job-blocking probability trade-off for Other Service Time Models}
We analyzed the most common service time models. Other distributions, such as bimodal and Weibull, are also interesting.
The optimal replication factor may behave differently for other light and heavy-tailed distributions since their computing cost versus\ time trade-offs are qualitatively different \cite{stragglers:AktasS19}. 

\subsubsection{System Service Rate for Blocking Systems}
System service rate is also a potential performance metric to evaluate the distributed system\cite{peng2021distributed}. We may analyze this metric for a general blocking system in the future, where the blocked job will be dropped rather than sent to the cloud in such a system. If the users require every computing result, they repeatedly send the request to the edge system until it is served. 

\bibliographystyle{IEEEtran}
\bibliography{bibli}

\begin{IEEEbiographynophoto}{Pei Peng} 
received his B.S. degree in information engineering from South China University of Technology, Guangzhou, China, in 2011, the M.S. degree in electronics and communication engineering from Shanghai Institute of Micro-system and Information Technology, Chinese Academy of Science, Shanghai, China, in 2014, and the PhD.\ degree in electrical and computer engineering from Rutgers, the State University of New Jersey, USA. 
Since 2022, he has been an assistant professor in the School of Telecommunication and Information Engineering at Nanjing University of Posts and Telecommunications, Nanjing, China. His research interests are coding in distributed systems, covert communications, edge computing, and satellite-terrestrial networks.
\end{IEEEbiographynophoto}

\begin{IEEEbiographynophoto}{Emina Soljanin}
 received the European Diploma in Electrical Engineering from the University of Sarajevo, Bosnia, in 1986, and the Ph.D. and M.S. degrees in electrical engineering from Texas A\&M University, College Station, TX, USA, in 1989 and 1994, respectively. She is
a Distinguished Professor of Electrical and Computer Engineering at Rutgers, The State University of New Jersey. Before moving to Rutgers, The State University of New Jersey, in January 2016, she was a member of Technical Staff (Distinguished) for 21 years in the Mathematical Sciences Research Center, Bell Labs. Her interests and expertise are broad and currently range from distributed computing to quantum information science. She has participated in numerous research and business projects, including power system optimization, magnetic recording, color
space quantization, hybrid ARQ, network coding, data and network security, distributed systems performance analysis, and quantum information theory. She was an outstanding alumnus at the Texas A\&M School of Engineering, the 2011 Padovani Lecturer, the 2016 and 2017 Distinguished Lecturer, and the 2019 IEEE Information Theory Society President
 \end{IEEEbiographynophoto}

\end{document}